\documentclass[10pt, conference]{IEEEtran} 

\usepackage{ltl}

\usepackage{tikz}
\usepackage{tikzscale}
\usepackage{pgfplots}
\pgfplotsset{
width=10cm,
compat=1.9}
\tikzset{every picture/.style=semithick,
}

\usetikzlibrary{positioning,mindmap,automata,fit,backgrounds,shapes, 
arrows,shapes.misc,patterns}
\usetikzlibrary{decorations.pathmorphing}
\usetikzlibrary{decorations.markings}

\usepackage{ltl}
\usepackage{tabulary}
\usepackage{url}
\usepackage{todonotes}
\usepackage{url}
\usepackage{cite}
\usepackage{xspace}
\usepackage{graphicx}
\usepackage{subfigure}
\usepackage{caption}
\usepackage{color}
\usepackage{relsize}
\usepackage{multicol}
\usepackage{hyphenat}
\usepackage{fancyvrb}
\usepackage{lipsum}
\usepackage{multirow}
\usepackage{xcolor, colortbl}

\newcommand{\MC}[2]{\mbox{\sf \small MC[#1, #2{}]}\xspace}
\newcommand{\LTL}{\mbox{LTL}\xspace}
\newcommand{\HyperLTL}{\mbox{HyperLTL}\xspace}

\newcommand{\alphabet}{\mathrm{\Sigma}}
\newcommand{\states}{\mathrm{\Sigma}}
\newcommand{\statespace}{\states}
\newcommand{\Trace}{\mathsf{Traces}}
\newcommand{\trace}{t}
\newcommand{\qtrace}{\eta}

\newcommand{\AP}{\mathsf{AP}}
\newcommand{\Next}{\X}

\newcommand{\Globally}{\G}
\newcommand{\V}{\mathcal{V}}

\newcommand{\States}{S}
\newcommand{\state}{s}
\newcommand{\trans}{\delta}
\newcommand{\kframe}{\mathcal{F}}
\newcommand{\krip}{\mathcal{K}}
\newcommand{\ktuple}{\langle S, s_\init, \trans, L \rangle}
\newcommand{\init}{\mathit{init}}

\newcommand{\tru}{\mathtt{true}}
\newcommand{\fals}{\mathtt{false}}
\newcommand{\quant}{\mathbb{Q}}

\newcommand{\comp}[1]{\textsf{\small #1}\xspace}

\newcommand{\GMNI}{\textsf{\small GMNI}\xspace}
\newcommand{\GNI}{\textsf{\small GNI}\xspace}

\newcommand\donotshow[1]{}
\newcommand{\F}{\LTLdiamond}
\newcommand{\G}{\LTLsquare}
\newcommand{\U}{\,\mathcal U\,}
\newcommand{\X}{\LTLcircle}
\newcommand{\Waitfor}{\,\mathcal W\,}
\newcommand{\suffix}[2]{#1[#2,\infty]}

\newcommand{\dec}{\mathit{dec}}
\newcommand{\ses}{\mathit{ses}}

\DefineVerbatimEnvironment{code}{Verbatim}{fontsize=\small}

\newtheorem{theorem}{Theorem}
\newtheorem{definition}{Definition}
\newtheorem{corollary}{Corollary}

\newcommand{\qed}{}

\renewcommand{\baselinestretch}{1.01}

\begin{document}

\title{The Complexity of Monitoring Hyperproperties}



\author{
\IEEEauthorblockN{Borzoo Bonakdarpour}
\IEEEauthorblockA{Department of Computer Science\\
Iowa State University, USA\\
Email: \sf{borzoo@iastate.edu}}
\and
\IEEEauthorblockN{Bernd Finkbeiner}
\IEEEauthorblockA{Reactive Systems Group\\
Saarland University, Germany\\
Email: \sf{finkbeiner@cs.uni-saarland.de}}
}

\maketitle

\begin{abstract}

We study the runtime verification of hyperproperties, expressed in the
temporal logic HyperLTL, as a means to inspect a system with respect
to security polices. Runtime monitors for hyperproperties analyze
trace logs that are organized by common prefixes in the form of a
tree-shaped Kripke structure, or are organized both by common prefixes
and by common suffixes in the form of an acyclic Kripke structure.
Unlike runtime verification techniques for trace properties, where the
monitor tracks the state of the specification but usually does not
need to store traces, a monitor for hyperproperties repeatedly model
checks the growing Kripke structure. This calls for a rigorous
complexity analysis of the model checking problem over tree-shaped and
acyclic Kripke structures.

We show that for {\em trees}, the complexity in the size of the Kripke 
structure is \comp{L-complete} independently of the number of quantifier 
alternations in the HyperLTL formula. For {\em acyclic} Kripke structures, the 
complexity is \comp{PSPACE-complete} (in the level of the polynomial hierarchy 
that corresponds to the number of quantifier alternations). The combined 
complexity in the size of the Kripke structure and the length of the HyperLTL 
formula is \comp{PSPACE-complete} for both trees and acyclic Kripke structures, 
and is as low as \comp{NC} for the relevant case of trees and alternation-free 
HyperLTL formulas. Thus, the size and shape of both the Kripke structure and
the formula have significant impact on the complexity of the model checking problem.

\end{abstract} 

\section{Introduction}
\label{sec:intro}

Most security properties related to confidentiality and information flow cannot 
be formulated as trace properties because they relate multiple computations. 
For example, \emph{observational determinism}~\cite{zm03} is satisfied if on 
every pair of computation traces where the observable inputs are the same, 
also the observable outputs are the same. This class of secure information 
flow policies has been characterized in a set-theoretic 
framework called \emph{hyperproperties}~\cite{cs10}. Hyperproperties can be 
expressed in the temporal logic HyperLTL~\cite{cfkmrs14}, which extends the 
linear-time temporal logic (LTL)~\cite{p77} with trace quantifiers and trace 
variables. Suppose, for example, that the observable input to a system is the 
atomic proposition~$i$ and the output is the atomic proposition~$o$. 
Observational determinism can then be expressed as the HyperLTL 
formula
\[
\varphi_{\mathsf{obs}} = \forall \pi.\ \forall \pi'.\ \G\, (i_\pi 
\Leftrightarrow i_{\pi'})\, 
\Rightarrow\,  \G\, (o_\pi \Leftrightarrow o_{\pi'}),
\]
where $\G$ is the usual ``globally'' operator of temporal logic: if two traces 
$\pi$ and $\pi'$ agree globally on $i$, then they must also globally
agree on $o$.

\begin{figure}
\centering
\subfigure[Set of linear traces.]{
    \includegraphics[scale=.55]{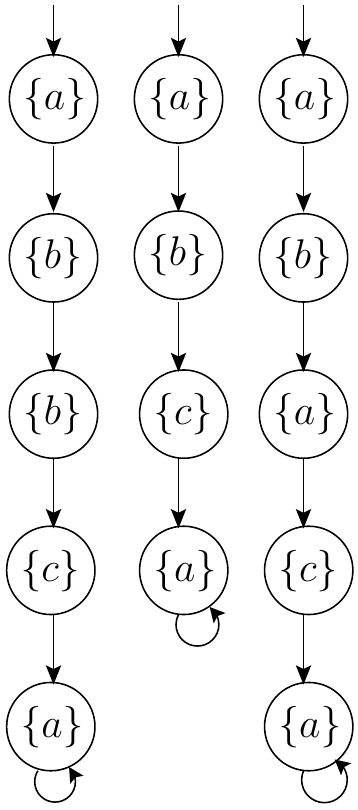}
    \label{fig:subfig1}
}
\hfill
\subfigure[Equivalent tree-shaped Kripke structure.]{
    \includegraphics[scale=.55]{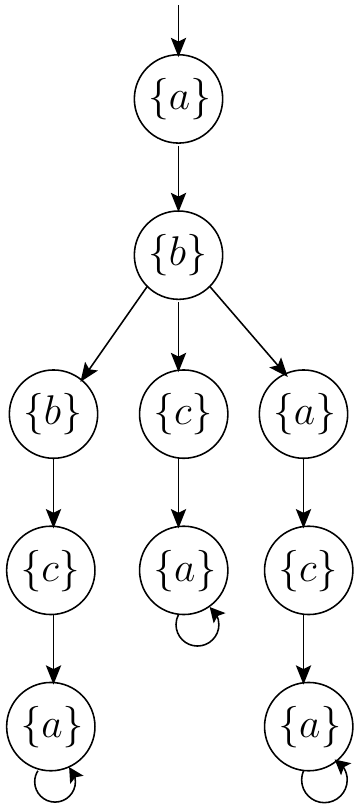}
    \label{fig:subfig2}
}
\hfill
\subfigure[Equivalent acyclic Kripke structure.]{
    \includegraphics[scale=.55]{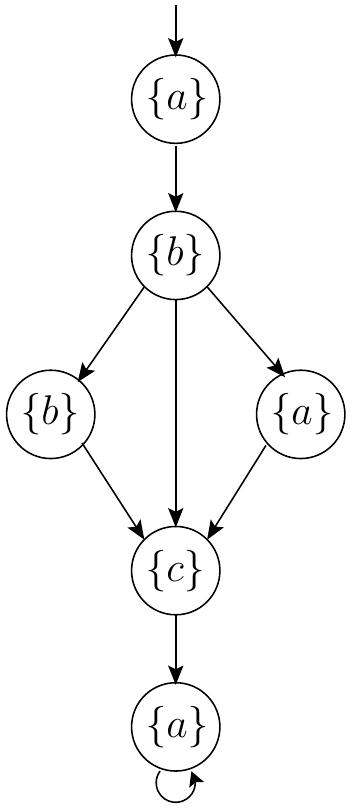}
    \label{fig:subfig3}
}
\caption{A trace log example and its assembly into space-efficient tree-shaped 
and acyclic Kripke structures.}
\label{fig:ks-shapes}
\end{figure}

{\em Runtime verification} is a technique that inspects the
health of a system by evaluating execution traces collected at run
time. Existing runtime verification techniques (e.g.,~\cite{gh01,klsss02,10.1007/978-3-642-04694-0_5,bls11})
evaluate a linear finite trace $t$ against a formula $\varphi$
expressed in a trace-based language such as LTL or regular
expressions.  Monitors for trace-based languages typically do not need
to record traces that are already evaluated. By contrast, a monitor
for hyperproperties must store a set $T$ of traces seen so far and
repeatedly check this growing set against the specification
(cf.~\cite{ab16,bsb17,fhst17}). For example, to monitor observational
determinism $\varphi_{\mathsf{obs}}$, the monitor has to examine every
existing {\em pair} of traces at all times and, hence, has to keep the
pairs that are already evaluated in a trace log. These trace logs may
be in the form of a simple linear collection of the traces seen so far
or, for space efficiency, organized by common prefixes and assembled
into a {\em tree-shaped} Kripke structure or by common prefixes as well as
suffixes assembled into an {\em acyclic} Kripke structure (see
Fig.~\ref{fig:ks-shapes}). Moreover, as a runtime monitor for
hyperproperties observes the execution traces while new traces (say
$T'$) are produced by the running system over time, the monitor has to
evaluate $\varphi$ with respect to $T \cup T'$ due to inter-trace assertions in
$\varphi$.  Over time, the size of the Kripke structure that
represents $T \cup T'$ may grow and its shape may change. Thus, a
fundamental research question is to study the complexity of the model checking problem for HyperLTL as the trace log grows over
time. For LTL, the complexity of the model checking problem for restricted Kripke
structures is known~\cite{kb11}. In particular, the model checking
problem is \comp{PSPACE-hard} (in the size of the formula) only if
there exists a strongly connected component with two distinct
cycles. For acyclic Kripke structures, the model checking problem is in
\comp{coNP}. If, additionally, the verification problem can be
decomposed into a polynomial number of finite path checking problems,
for example, if the Kripke structure is a tree or a directed graph
with constant depth, then the complexity reduces further to \comp{NC}.
Prior to our work, the complexity of the model checking
problem for hyperproperties over restricted Kripke structures was an open question.

\def\arraystretch{3}
\begin{table*}[t]
\centering
\scalebox{0.93}{\newcolumntype{K}[1]{>{\centering\arraybackslash}p{#1}}
\begin{tabular}{|K{2cm}||K{3cm}||K{3cm}|K{.5cm}||K{4.5cm}|}
\hline

& \multicolumn{3}{c||}{\cellcolor{black!15} \normalsize This paper} & \\
\cline{2-4} & {\bf Tree} & \multicolumn{2}{c||}{\bf Acyclic} & {\bf General}\\

\hline\hline
$\forall^+/\exists^+$ & \multirow{5}{*}{\em (Theorem~\ref{thm:Lcomp})} & 
\multicolumn{2}{c||}{\parbox[c]{2cm}{\centering \comp{NL-complete} \\
{\em (Theorem~\ref{thrm:af-flat})} }} & \comp{NL-complete}~\cite{frs15}\\
\cline{1-1}\cline{3-5}

$\exists^+\forall^+$/$\forall^+\exists^+$ &  \comp{L-complete} & 
\comp{NP/coNP-complete} & \multirow{3}{*}{\rotatebox[origin=c]{90}{\em 
(Theorem~\ref{thrm:eak-flat})}} & \comp{PSPACE-complete}~\cite{cfkmrs14}\\
\cline{1-1}\cline{3-3}\cline{5-5}

\multirow{3}{*}{$(\forall^*\exists^*)^*$}  &  & 
\comp{$\mathsf{\Pi}_{k}^p$-complete} & &
\multirow{2}{*}{\comp{$(k{-}1)$-EXPSPACE-complete}~\cite{markus}}\\
\cline{3-3}

 &  & $\mathsf{\Sigma}_{k}^p$\comp{-complete} & & \\
\cline{3-4}\cline{5-5}

 &  & 
\multicolumn{2}{c||}{\parbox[c]{3cm}{\centering \comp{PSPACE-complete} \\ {\em 
(Corollary~\ref{cor:aestar-s}) }}} & 
\comp{NONELEMENTARY}~\cite{cfkmrs14}\\
\hline

\end{tabular}
}
\vspace{2mm}
\caption{Complexity of the HyperLTL model checking problem in the size of 
the Kripke structure, where $k$ is the number of quantifier alternations in 
$(\forall^*\exists^*)^*$.}
\label{tab:system}
\end{table*}

\def\arraystretch{3.2}
\begin{table*}[t]
\centering
\scalebox{0.93}{\newcolumntype{K}[1]{>{\centering\arraybackslash}p{#1}}
\begin{tabular}{|K{2cm}||K{3cm}||K{3cm}|K{.5cm}||K{4.5cm}|}
    \hline

& \multicolumn{3}{c||}{\cellcolor{black!15} \normalsize This paper} & \\
\cline{2-4}
& {\bf Tree} & \multicolumn{2}{c||}{\bf Acyclic} & {\bf General}\\
\hline\hline

$\exists^k/\forall^k$ & \parbox[c]{2cm}{\centering
\comp{NC} \\ {\em (Theorem~\ref{thrm:aftree-c})} } & 
\multirow{2}{*}{\comp{NP/coNP-complete}} & 
\multirow{3}{*}{\rotatebox[origin=c]{90}{
\em (Theorem~\ref{thrm:aektree-c})}} & 
\multirow{2}{*}{\comp{PSPACE-complete}~\cite{frs15}}\\
\cline{1-2}

$\exists^+/\forall^+$ & \parbox[c]{3cm}{\centering
\comp{NP/coNP-complete} \\ {\em (Theorem~\ref{thrm:eak-tree})} } & & & \\
\cline{1-3}\cline{5-5}

$\exists\forall$/$\forall\exists$ & \parbox[c]{2cm}{\centering
\comp{NC} \\ {\em (Theorem~\ref{thrm:aetree-c})} } & 
\comp{$\mathsf{\Sigma^p_2/\Pi^p_2}$-complete} & 
& \comp{EXPSPACE-complete}~\cite{cfkmrs14}\\
\hline

\multirow{3}{*}{$(\forall^*\exists^*)^*$} & 
\multicolumn{2}{c|}{\comp{$\mathsf{\Pi}_{k+1}^p$-complete}}  & 
\multirow{2}{*}{\rotatebox[origin=c]{90}{
\em (Theorems~~\ref{thrm:eak-tree},\ref{thrm:aektree-c})}} &
\multirow{2}{*}{\comp{$k$-EXPSPACE-complete}~\cite{cfkmrs14}}\\
\cline{2-3}

 & \multicolumn{2}{c|}{\comp{$\mathsf{\Sigma}_{k+1}^p$-complete}} & & \\
\cline{2-5}

 & \multicolumn{2}{c}{\parbox[c]{3cm}{\centering \comp{PSPACE-complete} \\ {\em 
(Corollaries~\ref{cor:hltl-tree},\ref{cor:hltl-acyc})} }} & & 
\comp{NONELEMENTARY}~\cite{cfkmrs14}\\
\hline

\end{tabular}
}
\vspace{2mm}
\caption{Complexity of the HyperLTL model checking problem in the 
combined input, consisting of the Kripke structure and the HyperLTL formula, 
where $k$ is the number of quantifier alternations in $(\forall^*\exists^*)^*$.}
\label{tab:combined}
\end{table*}

\subsection{Contributions}

With this motivation, we study, in this paper, the impact of structural 
constraints on the complexity of the model checking problem for HyperLTL. As mentioned 
earlier, we are interested in Kripke structures that are tree-shaped or acyclic 
as two appropriate shapes to store execution trace logs. With respect to the 
HyperLTL formula, we are interested in the impact of the quantifier structure. 
Tables~\ref{tab:system} and \ref{tab:combined} summarize our new complexity 
results, contrasted with the known results for general Kripke 
structures~\cite{cfkmrs14,frs15,markus}, related to the equivalent model 
checking problem. Table~\ref{tab:system} shows the complexity of the model checking
problem in terms of the size of the Kripke structure alone. This \emph{system 
complexity} is often the most relevant complexity in practice, because the
system tends to be much larger than the specification. This is in particular 
true in runtime verification, where the Kripke structure that records the traces
seen so far grows over time, while the temporal formula remains the same.
Table~\ref{tab:combined} shows the \emph{combined complexity} in the full input, 
consisting of both the Kripke structure and the HyperLTL formula. Our results 
show that the shape of the Kripke structure plays a crucial role in the 
complexity of the model checking problem:

\begin{itemize}
\item {\bf Trees.} \ For trees, the complexity in the size of the Kripke 
structure is \comp{L-complete} independently of the number of quantifier 
alternations. The combined complexity in the size of the Kripke structure and 
the length of the HyperLTL formula is \comp{PSPACE-complete} (in the level of 
the polynomial hierarchy that corresponds to the number of quantifier 
alternations) and is as low as \comp{NC} for alternation-free fragment as well 
formulas of the form $\exists\forall$ and $\forall\exists$. 

\item {\bf Acyclic graphs.} \ For acyclic Kripke structures, the complexity is 
\comp{NL-complete} for the alternation-free fragment and is 
\comp{PSPACE-complete} for alternating formulas (in the level of the polynomial 
hierarchy that corresponds to the number of quantifier alternations). The 
combined complexity in the size of the Kripke structure and the length of the 
HyperLTL formula is also \comp{PSPACE-complete} in the level of the polynomial 
hierarchy that corresponds to the number of quantifier alternations.

\end{itemize}

\subsection{Significance of Contributions}

The significance of our results is multifold:

\begin{itemize}

\item Our results are in sharp contrast to the undecidability result 
of~\cite{fh16} and the non-elementary complexity of~\cite{cfkmrs14}, which 
has commonly been interpreted as suggesting that only the alternation-free fragment is worth considering in 
practical settings. Our results show that there is a lot that can be 
done for hyperproperties with alternations without exceeding \comp{PSPACE}.

\item An important observation from Tables~\ref{tab:system} 
and~\ref{tab:combined} is the impact of the shape of the Kripke structure and 
type of formula on the complexity. For example, the HyperLTL formula for Goguen 
and Meseguer's non-interference policy~\cite{gm82} is alternation-free for 
deterministic systems, while the same policy in a non-deterministic setting is 
of the form $\forall\forall\exists.\psi$, hence, one alternation. This changes 
the complexity from \comp{NL-complete} to \comp{coNP-complete} in acyclic 
graphs, while it remains \comp{L-complete} for trees. This shows
that there are trade offs, both in the choice of the shape of the trace logs and in the formula that 
represents the policy, with signficant practical implications. We will
present a more detailed motivating example on these trade offs in 
Section~\ref{sec:motive}.  

\item As discussed in~\cite{fhst17,bsb17}, monitoring hyperproperties may depend 
on the entire set of traces seen so far. This implies that a dependency on the 
total length and number of the traces is unavoidable. Having said that, our 
\comp{L-completeness} result for monitoring trees shows that the dependency in the total length is actually only
logarithmic. Also, if the complexity is measured in the length of the traces 
and the formula, our \comp{PSPACE-completeness} result shows that monitoring 
can be accomplished with a linear number of instances of the incremental 
traces. 

\item Our results are also of interest in the context of classic model checking.
In the
restricted Kripke structures, leaves in trees and acyclic graphs are defined 
to have self-loops, which encode infinite traces. Our results thus
have two applications: (1) classic model checking of restricted Kripke 
structures with infinite traces, and (2) runtime verification of a collected or 
evolving set of finite traces. Tree-shape and acyclic Kripke structures often 
occur as the natural representation of the state space of some protocols. For 
example, certain security protocols, such as authentication and session-based 
protocols (e.g., TLS, SSL, SIP) go through a finite sequence of \emph{phases}, 
resulting in an acyclic Kripke structure. The advantage of model checking 
restricted structures is particularly strong for HyperLTL formulas with many 
quantifier alternations: while the model checking problem over general Kripke 
structures cannot be solved by any elementary recursive 
function~\cite{cfkmrs14,frs15,markus}, the model checking problem for 
trees and acyclic graphs is in \comp{PSPACE}. The complexity in the size of a 
tree-shaped Kripke structure is even just \comp{L-complete}.

\end{itemize}
In a nutshell, we believe that the results in this paper provide the 
fundamental understanding of the runtime verification problem for secure 
information flow and pave the way for further research on efficient 
and scalable monitoring techniques.

\paragraph*{Organization} The remainder of this paper is organized as follows.
In Section~\ref{sec:prelim}, we review Kripke structures and HyperLTL. We 
present a detailed motivating example in Section~\ref{sec:motive}.
Section~\ref{sec:system} presents our results on the complexity of HyperLTL model checking 
in the size of the Kripke structure. Section~\ref{sec:combined} presents the 
results on the complexity in the combined input consisting of both 
the Kripke structure and the HyperLTL formula. We discuss related work in 
Section~\ref{sec:related}. Finally, we make concluding remarks in 
Section~\ref{sec:conclusion}.   

\section{Preliminaries}
\label{sec:prelim}

We begin with a quick review of Kripke structures and HyperLTL.

\subsection{Kripke Structures}
\label{subsec:krip}

Let $\AP$ be a finite set of {\em atomic propositions} and $\alphabet = 2^{\AP}$ 
be the {\em alphabet}. A {\em letter} is an element of $\alphabet$. A 
\emph{trace} $t$ over alphabet $\Sigma$ is an infinite sequence of letters in 
$\Sigma^\omega$:
$$t = t(0)t(1)t(2) \cdots~.$$


\begin{definition}
\label{def:kripke}

A {\em Kripke structure} is a tuple
$$\krip = \ktuple,$$
where 

\begin{itemize}
 \item $\States$ is a finite set of states;
 \item $\state_{\init} \in \States$ is the initial state;
 \item $\trans \subseteq \States \times \States$ is a transition relation, and 
 \item $L: S \rightarrow \statespace$ is a labeling function on the states of 
$\krip$.
\end{itemize}
We require that 
for each $\state \in \States$, there exists $\state' \in \States$, such 
that $(\state, \state') \in \trans$.

\end{definition}

For example, in Fig.~\ref{fig:kripke}, we have that $L(s_{init})= \{a\}, 
L(s_3)=\{b\}$, etc. The \emph{size} of the Kripke structure is the number of 
its states. The directed graph $\kframe = \langle \States, \trans \rangle$ is 
called the {\em Kripke frame} of the Kripke structure $\krip$. A {\em loop} in 
$\kframe$ is a finite sequence $\state_0\state_1\cdots \state_n$, such that 
$(\state_i, \state_{i+1}) \in \trans$, for all $0 \leq i < n$, and $(\state_n, 
\state_0) \in \trans$. We call a Kripke frame {\em acyclic}, if the only loops 
are self-loops on terminal states, i.e., on states that have no other outgoing 
transition. See Fig.~\ref{fig:kripke} for an example. Since 
Definition~\ref{def:kripke} does not allow terminal states, we only consider 
acyclic Kripke structures with such added self-loops.

We call a Kripke frame \emph{tree-shaped}, or, in short, a \emph{tree}, if 
every state $\state$ has a unique state $\state'$ with $(\state', \state) \in 
\trans$, except for the root node, which has no predecessor, and the leaf nodes, 
which, again because of Definition~\ref{def:kripke}, additionally have a 
self-loop but no other outgoing transitions.

A \emph{path} of a Kripke structure is an infinite sequence of states
$$\state(0)\state(1)\cdots \in \States^\omega,$$
such that:

\begin{itemize}

\item $\state(0) = \state_\init$, and

\item $(\state(i), \state({i+1})) \in \trans$, for all $i \geq 0$.

\end{itemize}
A trace of a Kripke structure is a trace $t(0)t(1)t(2) \cdots \in \Sigma^\omega$ 
such that there exists a path $\state(0)\state(1)\cdots \in \States^\omega$ with 
$t(i) = L(\state(i))$ for all $i\geq 0$. We denote by $\Trace(\krip, \state)$ 
the set of all traces of $\krip$ with paths that start in state $\state \in 
\States$.

In the context of monitoring, we assume that traces of a system under inspection 
are given as a tree-shaped or acyclic Kripke structure. These type of Kripke 
frames are obviously more space efficient than a set of linear traces, because
trees allow us to organize the traces according to common prefixes and acyclic graphs according to both common prefixes and common suffixes.

\begin{figure}[t]
\centering
\begin{tikzpicture}[-,>=stealth',shorten >=.5pt,auto,node distance=2cm, 
semithick, initial text={}]

\node[initial, state] [text width=1em, text centered, minimum 
  height=2.5em](0) {\hspace*{-1.25mm}$\{a\}$};

\node [above left = 0.005 cm and 0.1 cm of 0](label){$s_{init}$};


\node[state, above right=of 0][text width=1em, text centered, minimum 
  height=2.5em] (1) {\hspace*{-1.25mm}$\{a\}$};

\node [above left = 0.005 cm and 0.1 cm of 1](label){$s_{1}$};

\node[state, right=of 1][text width=1em, text centered, minimum 
height=2.5em] (2) {\hspace*{-1.25mm}$\{b\}$};

\node [above left = 0.005 cm and 0.1 cm of 2](label){$s_{2}$};

\node[state, right=of 0][text width=1em, text centered, minimum 
height=2.5em] (3) {\hspace*{-1.25mm}$\{b\}$};

\node [above left = 0.005 cm and 0.1 cm of 3](label){$s_{3}$};

\draw[->]   
  (0) edge node (01 label) {} (1)
  (1) edge node (12 label) {} (3)
  (0) edge node (03 label) {} (3)
  (1) edge node (12 label) {} (2)
  (2) edge [loop right] node (22 label) {} (2)
  (3) edge [loop right] node (33 label) {} (3);
    
\end{tikzpicture}\caption {Example of an acyclic Kripke structure (with 
self-loops at otherwise terminal states).}
\label{fig:kripke}
\end{figure}
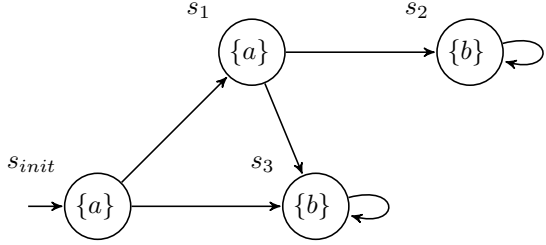

\subsection{HyperLTL}
\label{subsec:hltl}


%
\HyperLTL~\cite{cfkmrs14} is a temporal logic for expressing hyperproperties.
A hyperproperty~\cite{cs10} is a set of sets of execution traces.
\HyperLTL generalizes \LTL by allowing explicit quantification over multiple execution 
traces simultaneously. The set of \HyperLTL formulas is defined inductively by 
the following grammar:
\begin{equation*}
\begin{aligned}
& \varphi ::= \exists \pi . \varphi \mid \forall \pi. \varphi \mid \phi \\
& \phi ::= \tru \mid a_\pi \mid \lnot \phi \mid \phi \vee \phi \mid \phi \
\U \, \phi \mid \Next \phi 
    \end{aligned}
\end{equation*}
where $a \in \AP$ is an atomic proposition and $\pi$ is a {\em trace
variable} from an infinite supply of variables $\V$. The Boolean
connectives $\neg$ and $\vee$ have the usual meaning, $\U$ is the
temporal \emph{until} operator and $\X$ is the temporal \emph{next}
operator. We also consider the usual derived Boolean connectives, such
as $\wedge$, $\Rightarrow$, and $\Leftrightarrow$, and the derived
temporal operators \emph{eventually} $\F\varphi\equiv \tru
\,\U\,\varphi$, \emph{globally} $\G\varphi\equiv\neg\F\neg\varphi$,
and \emph{weak until}: $\varphi \Waitfor \psi \equiv (\varphi\, \U
\psi) \vee \G \varphi$. The quantified formulas $\exists \pi$ and
$\forall \pi$ are read as `along some trace $\pi$' and `along all
traces $\pi$', respectively.  A \emph{sentence} is a closed formula,
i.e., the formula that has no free trace variables. A formula with
only universal or only existential quantifiers is called
\emph{alternation-free}. Such formulas have \emph{alternation depth}
0.  The alternation depth of formulas with both existential and
universal quantifiers is the number of alternations from existential
to universal quantifiers and from universal to existential
quantifiers.

The semantics of HyperLTL is defined with respect to a trace assignment, a 
partial mapping~$\Pi \colon \V \rightarrow \Sigma^\omega$. The assignment with 
empty domain is denoted by $\Pi_\emptyset$. Given a trace assignment~$\Pi$, a 
trace variable~$\pi$, and a trace~$t$, we denote by $\Pi[\pi \rightarrow t]$ 
the assignment that coincides with $\Pi$ everywhere but at $\pi$, which is 
mapped to $t$. Furthermore, $\suffix{\Pi}{j}$ denotes the assignment mapping 
each trace~$\pi$ in $\Pi$'s domain to
$$\Pi(\pi)(j)\Pi(\pi)(j+1)\Pi(\pi)(j+2)\cdots~.$$
The satisfaction of a HyperLTL formula $\varphi$ over a trace assignment $\Pi$ 
and a set $T$ of traces, denoted by $T,\Pi \models \varphi$, is defined as 
follows:
\[
\begin{array}{l@{\hspace{1em}}c@{\hspace{1em}}l}
  T, \Pi \models a_\pi & \text{iff} & a \in \Pi(\pi)(0),\\
  T, \Pi \models \neg \phi & \text{iff} & T, \Pi \not\models \phi,\\
  T, \Pi \models \phi_1 \vee \phi_2 & \text{iff} & T, \Pi \models \phi_1 \text{ 
or } T, \Pi \models \phi_2,\\
  T, \Pi \models \X \phi & \mbox{iff} & T,\suffix{\Pi}{1} \models \phi,\\
  T, \Pi \models \phi_1 \U \phi_2 & \text{iff} &  \exists i \ge 0: 
T,\suffix{\Pi}{i} \models \phi_2 \; \wedge\ \vspace*{-0.25cm}\\
  & &  \forall j \in [0, i): T,\suffix{\Pi}{j} \models \phi_1,\\
  T, \Pi \models \exists \pi.\ \varphi & \text{iff} & \exists t \in T: 
T,\Pi[\pi \rightarrow t] \models \varphi,\\
  T, \Pi \models \forall \pi.\ \varphi & \text{iff} & \forall t \in T: 
T,\Pi[\pi\rightarrow t] \models \varphi.
  \end{array}
\]

We say that a set $T$ of traces satisfies a sentence~$\varphi$, denoted by $T 
\models \varphi$, if $T, \Pi_\emptyset \models \varphi$. A Kripke structure 
$\krip = \ktuple$ satisfies a HyperLTL formula $\varphi$, denoted by $\krip 
\models \varphi$, iff $\Trace(\krip,\state_{\init}) \models \varphi$.\\


\noindent \emph{Example.} Consider the HyperLTL formula
$$\varphi =  \forall \pi_{1}. \forall \pi_{2}.\,  a_{\pi_1} \U b_{\pi_2}$$
and the Kripke structure in Fig.~\ref{fig:kripke}. The Kripke structure does not 
satisfy $\varphi$. For example, the trace assignment $\Pi$ that assigns to 
$\pi_1$ the trace $\{a\}\{b\}^\omega$  and to $\pi_2$ the trace  
$\{a\}\{a\}\{b\}^\omega$ does not satisfy $a_{\pi_1} \U b_{\pi_2}$.\\

Standard \emph{linear-time temporal logic} (LTL) is the fragment of HyperLTL 
with a single quantifier. Typically, the quantifier is universal and is left 
implicit, i.e., the LTL formula $\varphi = \forall \pi.\ \psi$ is written as 
$\psi$ with the index $\pi$ omitted from all atomic propositions. We say that a 
trace $t$ satisfies an LTL formula $\varphi$, denoted by $t \models \varphi$, if 
$\{t\} \models \varphi$.

We note that although our focus in this paper is on runtime verification 
(hence, a finite number of finite traces), for simplicity and without loss of 
generality, we use the infinite semantics of HyperLTL. To this end, we assume 
that the leaves of Kripke frames have self-loops that corresponds to the 
``stuttering'' semantics of finite-trace temporal logics.

\begin{figure*}[t]
\centering
\includegraphics[width=	\textwidth]{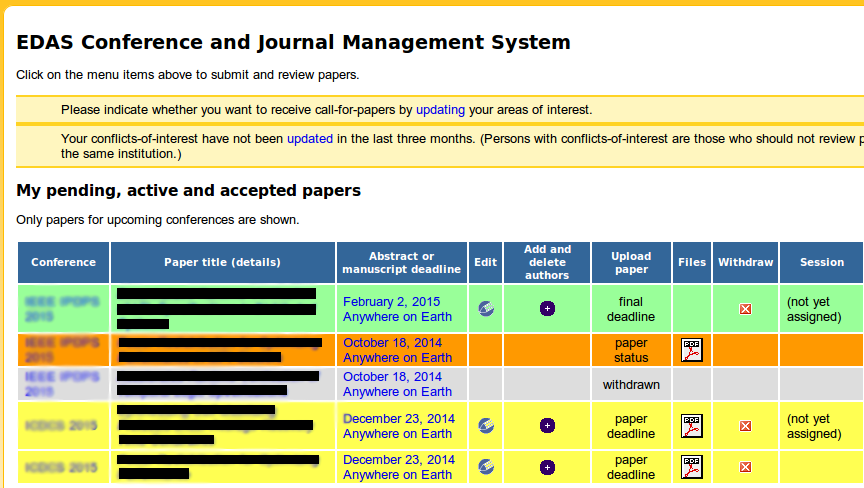}
\caption{EDAS conference management website's information leak.}
\label{fig:edas}
\end{figure*}

\section{Motivating Example}
\label{sec:motive}

\subsection{EDAS Conference Manager Bug}

\newcommand{\acc}{\mathsf{acc}}
\newcommand{\rej}{\mathsf{rej}}
\newcommand{\undec}{\mathsf{undec}}

We demonstrate the importance of the problem under investigation in this paper
with a real-life information leak encountered by the first author while using 
the EDAS Conference Management System\footnote{\url{http://www.edas.info}}. 
Fig.~\ref{fig:edas} shows an anonymized screenshot of
the EDAS web interface~\cite{ab16}. The color-coded table displays the status of 
submitted papers by the user: accepted (green), rejected (orange), withdrawn 
(grey), and pending (yellow). Now, consider the well-known Goguen and 
Meseguer's {\em non-interference} (\GMNI) security policy~\cite{gm82} for {\em 
deterministic} systems, where a low-privileged user (in this case, the author) 
should not be able to acquire any information about the activities (if any) of 
the high-privileged user (in this case, the conference PC chair). The HyperLTL 
formula for this policy in the context of our example is the following:
\begin{align*}
\varphi_{\mathsf{GMNI}} = \forall \pi. \forall \pi' . \Big(\Globally (\dec = 
\lambda)_{\pi'} \, \land \, & \Globally (\dec_{\pi} \neq \dec_{\pi'})\Big) \,
\Rightarrow\\
& \Globally\Big(\ses_{\pi} \Leftrightarrow \ses_{\pi'}\Big)
\end{align*}
where high input variable $\dec$ for a submission, ranging over $\{\acc, \rej, 
\undec\}$, contains the internal decision of the conference chair for the 
submission and low output proposition $\ses$ represents whether or not the 
submission is assigned to a session for presentation. By abuse of notation, we 
denote the value of variable $\dec$ in the associated state of trace $\pi$ by 
$\dec_\pi$. Finally, $\lambda$ denotes an arbitrary dummy value for the $\dec$. 

The web interface exhibits the following blunt violation of \GMNI, i.e., the 
author can learn the internal decision of the chair, while the status of the 
paper is pending. The first two rows show the status of two papers submitted to 
a conference after their notification (i.e., values sent on the low-observable 
channel): the first paper is accepted while the second is rejected. The last 
two rows show two other papers submitted to a different conference whose status 
are pending at the time the screenshot is taken. Although the authors should 
not be able to infer the internal decision making activities (i.e., high 
inputs) of the conference chair before the notification, this table leaks these 
activities as follows. When the chair sets $\dec = \acc$, the paper is supposed 
to be assigned to a session in the technical program, while a rejected paper 
(i.e., $\dec = \rej$) does not need to be assigned to a session. Now, by 
comparing the rows, one can observe that their `Session' column have the 
same value (i.e., `not yet assigned'). Likewise, the second and the last rows 
have an empty `Session' column. This simply means that the table reveals the 
internal status of the fourth and last papers as accepted and 
rejected, respectively, although their external status are pending. More 
specifically, in formula $\varphi_{\mathsf{GMNI}}$, if $\pi'$ and $\pi$ are 
instantiated by the last two yellow rows, respectively, then purging $\dec$ 
by $\lambda$ in $\pi'$ will result in different $\ses$ observations, which 
clearly is a violation of non-interference through the four independent 
executions to generate the HTML table rows\footnote{We note that EDAS has 
fixed this bug after we brought it to their attention.}.

\subsection{The Need for Runtime Monitoring}

The above example illustrates how a security policy can easily be 
violated due to a careless implementation, where the value of high variable 
$\dec$ flows in the publicly-observable variable $\ses$, although the chair did 
not take any inappropriate action that directly violates the security policy. 
This example demonstrates the need for designing techniques for monitoring the 
functional as well as security aspects of systems such as an online conference 
manager to inspect their health at run time or through periodic offline trace 
log analysis.

A key step in deploying any type of verification is identifying the 
specification of the system in terms of a formula. For our conference manager 
system, we now identify the specification in a sequence of steps starting from 
a simple formula, which is evolved into more complex ones: 

\begin{itemize}

\item If the specification is only concerned with monitoring non-interference in 
deterministic executions, then formula $\varphi_{\mathsf{GMNI}}$ suffices. In 
this case, according to Table~\ref{tab:system}, the complexity of monitoring a 
tree-shaped (respectively, acyclic) trace log is \comp{L-complete} 
(respectively, \comp{NL-complete}) in the size of the log.

\item Next, let us imagine, the generation of the HTML report is accomplished 
by a set of concurrent threads. In this case, we need to refine 
$\varphi_\mathsf{GMNI}$ to obtain a stronger notion of confidentiality known as 
the Generalized Non-interference (\GNI)~\cite{m88}, which permits 
nondeterminism in the behavior, but stipulates that low-security outputs may 
not be altered by the injection of high-security inputs:
\begin{align*}
\varphi_{\mathsf{GNI}} = \forall \pi. \forall \pi'. \exists \pi''. 
& \G(\dec_\pi = \dec_{\pi''}) \; \wedge \\
& \G (\ses_{\pi'} \Leftrightarrow \ses_{\pi''})
\end{align*}
The trace $\pi''$ is an interleaving of the high inputs trace $\pi$ and the low 
outputs of the trace $\pi'$. In this case, according to Table~\ref{tab:system}, 
the complexity of monitoring a tree-shaped (respectively, acyclic) trace log is 
\comp{L-complete} (respectively, \comp{coNP-complete}) in the size of the log. 
As can be seen, in case of asyclic trace logs, there is a significant jump in 
the complexity hierarchy of monitoring.

\item As mentioned earlier, the EDAS information leak was due to an 
implementation bug, rather than by a mistake by the chair. For the designer of 
the conference management system, this is an important distinction:

\begin{itemize}
\item Information leaks caused by an incorrect implementation should be fixed 
by eliminating the bug in the implementation, and

\item Information leaks caused by the user could be fixed by educating the user 
or by improving the user interface, for example by issuing  an explicit warning, 
or might not even need fixing if the information leak was intentional. 

\end{itemize}

In the next step, we will further refine the specification to \emph{only} 
refer to information leaks that are due to errors in the implementation, 
ignoring information leaks that are caused by the conference chair. For this 
purpose, we specify that a trace $\pi_1$ is OK (from the system 
implementation's point of view) even if $\pi_1$ results in a leak, as long as 
there exists a trace $\pi_2$, representing a different interaction of the chair 
with the system, that avoids the leak. In order to prevent trivial 
alternatives, such as ``do nothing'', we only consider alternative user 
behaviors that would accomplish the same objectives. Now, suppose that for two 
traces $\pi_1$ and $\pi_2$, the predicate $\mathit{obj}(\pi_1, \pi_2)$ 
indicates that $\pi_1$ and $\pi_2$ accomplish the same functional objectives, 
e.g., 
\begin{align*}
& \mathit{obj}(\pi_1, \pi_2)=\\
& \Big(\F (\dec = \acc)_{\pi_1} \, \Rightarrow \,\F (\dec = \acc)_{\pi_2}\Big) 
\; \wedge\\
&\Big(\F (\dec = \rej)_{\pi_1} \;\; \Rightarrow \,\F (\dec = \rej)_{\pi_2}\Big)
\end{align*}
i.e., if $\pi_1$ prescribes that a paper is accepted (respectively, 
rejected), then the same decision is made for the paper in $\pi_2$ as well. 
Then, our refined property is expressed by the following HyperLTL formula:
\begin{align*}
\varphi_{\mathsf{ref}} =  \forall \pi_1. \exists \pi_2. \forall \pi_3. \exists 
\pi_4. \; & obj(\pi_1, \pi_2) \; \wedge \\
& \G(\dec_{\pi_4} = \dec_{\pi_2}) \; \wedge \\
& \G (\ses_{\pi_4} \Leftrightarrow \ses_{\pi_3})
\end{align*}
The formula (with three quantifier alternations) expresses noninterference with 
the modification that the universally quantified trace $\pi_1$ is replaced by a 
existentially quantified trace $\pi_2$ that satisfies the same objectives. 
According to Table~\ref{tab:system}, the complexity of monitoring the refined 
property in a tree-shaped (respectively, acyclic) trace log is \comp{L-complete}
(respectively, \comp{$\mathsf{\Pi}_{4}^p$-complete}) in the size of the 
Kripke structure. If we did not allow non-determinism, which translates to 
removing the innermost existential quantifier and, hence, one less 
alternation in the formula quantifiers, the complexities would be 
\comp{L-complete} and \comp{$\mathsf{\Pi}_{3}^p$-complete}, respectively.

\end{itemize}

As can be seen in this example, the choice of the shape of the Kripke structure 
and the HyperLTL formula play a crucial role in the complexity of the model checking 
problem. This observation motivates rigorously investigating the complexity of 
RV for tree-shaped and acyclic Kripke structures. Our findings, summarized in 
Tables~\ref{tab:system} and~\ref{tab:combined}, are presented in detail in 
Sections~\ref{sec:system} and \ref{sec:combined}, respectively.

\section{System Complexity}
\label{sec:system}

In this section, we analyze the complexity of the model checking problem in the size of the 
Kripke structure. We use the following notation to distinguish the different 
variations of the problem:
\begin{center}
 \MC{\sf Fragment}{\mbox{\sf Frame Type}},
\end{center}
where

\begin{itemize}
 \item \comp{MC} is the model checking problem, i.e., the problem to 
determine whether or not $\krip \models \varphi$, where $\krip$ is a Kripke 
structure and $\varphi$ is a closed HyperLTL formula;
 
\item \comp{Fragment} is one of the following for $\varphi$:

\begin{itemize}
 \item \comp{AF-HyperLTL} refers to the alternation-free fragment of HyperLTL 
(i.e., $\exists^+\psi$ or $\forall^+\psi$);

\item \comp{(EA)$k$-HyperLTL}, for $k\geq 0$, denotes the fragment with $k$ 
alternations and a lead existential quantifier, where $k=0$ means an 
alternation-free formula with only existential quantifiers;

\item \comp{(AE)$k$-HyperLTL}, for $k\geq 0$, denotes the fragment with $k$ 
alternations and a lead universal quantifier, where $k=0$ means an 
alternation-free formula with only universal quantifiers,

\item \comp{HyperLTL} is the full logic HyperLTL, and

\end{itemize} 

\item \comp{Frame Type} is either \comp{tree}, \comp{acyclic}, or 
\comp{general}.

\end{itemize}


\subsection{Tree-shaped Graphs}

Our first result is that the model checking problem for tree-shaped Kripke structures is 
\comp{L-complete} in the size of the Kripke structures. This result is 
particularly interesting, as system trace logs are very often stored as a set 
of traces grouped by common prefixes. 



\begin{theorem}
  \label{thm:Lcomp}
\MC{HyperLTL}{\mbox{tree}} is \comp{L-complete} in the size of the Kripke structure.
\end{theorem}

\begin{proof}
For the upper bound, we note that the number of traces in a tree is bounded by 
the number of states, i.e., the size of the Kripke structure. The model checking algorithm 
maintains for each trace variable a counter on the number of traces, i.e., a 
logarithmic number of bits in size of the Kripke structure. To evaluate the 
inner LTL subformula, determine, in a backwards fashion, whether a subformula 
holds for a particular trace position. We need two counters on the 
length of the trace (corresponding to the variables $i$ and $j$ in the 
semantics of Until) for each Until subformula. Since the length of the trace is 
again bounded by the number of states, again a logarithmic number of bits will 
suffice. (Note that, since we are only interested in the complexity in the size 
of the Kripke structure, we consider the number of subformulas to be constant.)

The lower bound follows from the \comp{L-hardness} of 
ORD~\cite{et97}. ORD is the graph-reachability problem for directed 
line graphs. Graph reachability from $s$ to $t$ can be expressed with the 
formula $\exists \pi .\ \F (s_\pi \wedge \F t_\pi)$.
\qed
\end{proof}
 
\subsection{Acyclic Graphs}

We now turn to acyclic graphs. Acyclic Kripke structures are interesting in two 
contexts: (1) efficient storage of system trace logs in runtime verification, grouping the traces according to common prefixes 
\emph{and} common suffixes, and (2) analyzing certain security protocols, in particular
authentication algorithms, which often consist of sequences of phases 
with no repetitions or loops. Such applications result in acyclic 
Kripke structures. We develop results for three different fragments of HyperLTL: (1) the alternation-free 
fragment (Theorem~\ref{thrm:af-flat}), the bounded-alternation fragment 
(Theorem~\ref{thrm:eak-flat}), and (3) full HyperLTL 
(Corollary~\ref{cor:aestar-s}).

\subsubsection{Alternation-free Formulas}



\begin{theorem}
\label{thrm:af-flat}
\MC{AF-HyperLTL}{\mbox{acyclic}} is \comp{NL}-complete in the size of the 
Kripke structure.
\end{theorem}

\begin{proof}
For the upper bound, we consider the case that the HyperLTL formula is 
existential, i.e., it is of the form
$$\exists \pi_1 \ldots \exists \pi_k .\, \varphi,$$
where $\varphi$ does not contain any trace quantifiers. For the case 
that the formula is universal, i.e., it is of the form
$$\forall \pi_1 \ldots \forall \pi_k .\, \varphi,$$
we check the formula $\exists \pi_1 \ldots \exists \pi_k .\, \neg \varphi$
and report the complemented result.

We consider the self-composition of the Kripke structure. Let $\krip = \ktuple$ 
be a  Kripke structure, and let  $\exists \pi_1 \ldots \exists \pi_k .\, 
\varphi$ be an existential HyperLTL formula. The \emph{self-composition} of 
$\krip$ is the Kripke structure $\krip' = \langle S^k,s_\init^k, \trans', L' 
\rangle$, where 
\begin{align*}
& S^k = \overbrace{S \times S \times \dots \times S}^{k~\text{times}} \\
& s_\init^k = \overbrace{(s_\init, s_\init, \dots, s_\init)}^{k \text{ times 
}}\\
& \trans' = \Big\{ \big((s_1, \ldots, s_k), (s'_1, \ldots, s'_k)\big) \mid 
\forall i \in [1, k]: (s_i, s_i') \in \trans\Big\}\\
& L'(s_1, \ldots, s_k) = \Big\{ a_i \mid \exists i \in [1, k]: a \in 
L(s_i)\Big\}.
\end{align*}
It is easy to see that the self-composition of an acyclic Kripke structure is 
again acyclic.

For the HyperLTL formula $\exists \pi_1 \ldots \exists \pi_k .\, \varphi$, let 
$\varphi'$ be the same as inner LTL formula $\varphi$, where every indexed 
proposition $a_{\pi_i}$, for some $i \in [1, k]$, is replaced by the atomic 
proposition $a_i$. Now, the Kripke structure $\krip$ satisfies $\exists \pi_1 
\ldots \exists \pi_k .\, \varphi$, iff there is a path in the self-composition 
$\krip'$, such that the corresponding trace satisfies $\varphi'$. Since the 
Kripke structure is acyclic, the length of the traces is 
bounded by the number of states of the Kripke structure. We can, therefore, 
nondeterministically guess the trace that satisfies $\varphi'$, using a counter 
with a logarithmic number of bits in the number of states of $\krip$. 

The lower bound follows from the \comp{NL-hardness} of the graph-reachability 
problem for ordered graphs~\cite{LENGAUER199263}. Ordered graphs are acyclic 
graphs with a vertex numbering that is a topological sorting of the vertices. 
As in the proof of Theorem~\ref{thm:Lcomp}, we express graph reachability from 
$s$ to $t$ with the formula $\exists \pi .\ \F (s_\pi \wedge \F t_\pi)$. \qed
\end{proof}

\subsubsection{Formulas with Bounded Alternation Depth}
$ $ 

Next, we consider formulas where the number of quantifier alternations is bounded by a 
constant $k$. We show that changing the frame structure from a tree to an 
acyclic graph results in significant increase in complexity (see 
Table~\ref{tab:system}).

\begin{theorem}
\label{thrm:eak-flat}
\MC{\mbox{(EA)}$k$\mbox{-HyperLTL}}{\mbox{acyclic}} is  
\comp{$\mathsf{\Sigma^p_{k}}$-complete} in the size of the Kripke structure. 
\MC{\mbox{(AE)}$k$\mbox{-HyperLTL}}{\mbox{acyclic}} is 
\comp{$\mathsf{\Pi^p_{k}}$-complete} in the size of the Kripke structure.

\end{theorem}

\begin{proof}
We show membership in \comp{$\mathsf{\Sigma^p_{k}}$} and 
\comp{$\mathsf{\Pi^p_{k}}$}, respectively, by induction over $k$. According to 
Theorem~\ref{thrm:af-flat}, the model checking problem for $k = 0$, where the formula is 
alternation-free, is solvable in polynomial time. For $k+1$ quantifier 
alternations, suppose that the first quantifier is existential. Since the Kripke 
structure is acyclic, the length of the traces is bounded by the number of 
states. We can thus nondeterministically guess the existentially quantified 
traces in polynomial time and then verify the correctness of the guess, by the 
induction hypothesis, in \comp{$\mathsf{\Pi^p_{k}}$}. Hence, the model checking problem for 
$k+1$ is in \comp{$\mathsf{\Sigma^p_{k+1}}$}. Likewise, if the first quantifier 
is universal, we universally guess the universally quantified traces in 
polynomial time and verify the correctness of the guess, by the induction 
hypothesis, in \comp{$\mathsf{\Sigma^p_{k}}$}. Hence, the problem of 
determining $\krip \models \varphi$ for $k+1$ alternations in $\varphi$ is in 
\comp{$\mathsf{\Pi^p_{k+1}}$}. 

For the lower bound, we show that the model checking problem for HyperLTL formula with $k$ 
alternations is \comp{$\mathsf{\Sigma^p_{k}}$-hard} and 
\comp{$\mathsf{\Pi^p_{k}}$-hard}, respectively, via a reduction from the {\em 
quantified Boolean formula} (QBF) satisfiability problem~\cite{gj79}:

\begin{quote}

{\em Given is a set of Boolean variables, $\{x_1, x_2, \dots, x_n\}$, and a 
quantified Boolean formula
$$y=\quant_1 x_1.\quant_1 x_2\dots\quant_{n-1} x_{n-1}.\quant_n x_n.(y_1 \, 
\wedge \, y_2 \, \wedge \dots \wedge \, y_m)$$
where each $\quant_i \in \{\forall, \exists\}$ ($i \in [1, n]$) and each clause 
$y_j$ ($j \in [1, m]$) is a disjunction of three literals (3CNF). Is $y$ 
true?}
 
\end{quote}
If $y$ is restricted to at most $k$ alternations of quantifiers, then QBF 
satisfiability is complete for \comp{$\mathsf{\Sigma^p_{k+1}}$} if $\quant_1 = 
\exists$, and for \comp{$\mathsf{\Pi^p_{k}}$} if $\quant_1 = \forall$. We note 
that in the given instance of the QBF problem: 

\begin{itemize}

\item The clauses may have more than three literals, but three is sufficient of 
our purpose;

\item The inner Boolean formula has to be in conjunctive normal form in order 
for our reduction to work;

\item Without loss of generality, the variables in the literals of the same 
clause are different (this can be achieved by a simple pre-processing of the 
formula), and

\item If the formula has $k$ alternations, then it has $k+1$ alternation {\em 
depths}. For example, formula
$$\forall x_1.\exists x_2. (x_1 \vee \neg x_2)$$ 
has one alternation, but two alternation depths: one for $\forall x_1$ and the 
second for $\exists x_2$. By $d(x_i)$, we mean the alternation depth of Boolean 
variable $x_i$.

\end{itemize}

We now present a mapping from an arbitrary instance of QBF with $k$ 
alternations and where $\quant_1 = \exists$ to the model checking problem of an acyclic 
Kripke structure and a HyperLTL formula with $k$ quantifier alternations. Then, 
we show that the Kripke structure satisfies the HyperLTL formula if and only if 
the answer to the QBF problem is affirmative. Figures~\ref{fig:mapping-sys} 
and~\ref{fig:qbfmodel-sys} show an example.

\noindent \textbf{Kripke structure $\krip = \ktuple$: } 

\begin{itemize}

\item {\em (Atomic propositions $\AP$)} For each alternation depth $d \in [1, 
k+1]$, we include an atomic proposition $q^d$. We furthermore include three 
atomic propositions: $c$ is used to mark the clauses, $p$ is used to force 
clauses to become true if a Boolean variable appears in a clause, and 
proposition $\bar{p}$ is used to force clauses to become true if the negation 
of a Boolean variable appears in a clause in our reduction. Thus,
$$\AP = \big\{c, p, \bar{p}\big\} \; \cup \; \big\{ q^d \mid d \in [1, 
k+1]\big\}.$$

\item {\em (Set of states $\States$)}  We now identify the members of $\States$:

\begin{itemize}

\item First, we include an initial state $\state_\init$ and a state $r_0$. 
Then, for each clause $y_j$, where $j \in [1, m]$, we include a state 
$r_j$, labeled by proposition $c$.
  
\item For each clause $y_j$, where $j \in [1, m]$, we introduce the 
following $2n$ states: 
$$\Big\{v^j_i, u^j_i \mid i \in [1, n]\Big\}.$$
Each state $v^j_i$ is labeled with propositions $q^{d(x_i)}$, and with $p$ if 
$x_i$ is a literal in $y_j$, or with $\bar{p}$ if $\neg x_i$ is a literal in 
$y_j$.

\item For each Boolean variable $x_i$, where $i \in [1, n]$, we include three 
states $s_i$, $\bar{s}_i$, and $\hat{s}_i$. Each state $s_i$ (respectively, 
$\bar{s}_i$) is labeled by $p$ and $q^{d(x_i)}$ (respectively, $\bar{p}$ and
$q^{d(x_i)}$).

\end{itemize}
Thus,
\begin{align*}
S = & \big\{s_\init \big\} \, \cup \, \big\{r_j \mid j \in [0, m]\big\}  \; 
\cup \\
& \big\{v^j_i, u^j_i, s_i, \bar{s_i}, \hat{s}_i \mid i \in [1, n] \wedge 
j \in [1,m]\big\}.
\end{align*}

\item {\em (Transition relation $\trans$)} We now identify the members of 
$\trans$:
 
\begin{itemize}

\item We include a transition $(\state_\init, r_j)$, for each $j \in [0, m]$.
 
\item We add transitions $(r_j, v^j_1)$ for each $j \in [1, m]$.

\item For each $i \in [1, n]$ and $j \in [1, m]$, we include transitions 
$(v^j_i, u^j_i)$. For each $i \in [1, n)$ and $j \in [1, m]$, we include 
transitions $(u^j_i, v^j_{i+1})$.

\item For each $i \in [1, n]$, we include transitions $(s_i, \hat{s}_i)$ and 
$(\bar{s}_i, \hat{s}_i)$. For each $i \in [1, n)$, we include transitions 
$(\hat{s}_i, s_{i+1})$ and $(\hat{s}_i, \bar{s}_{i+1})$.

\item We include two transitions $(r_0, s_1)$ and $(r_0, \bar{s}_1)$.

\item Finally, we include self-loops $(\hat{s}_n, \hat{s}_n)$ and $(u_n^j, 
u_n^j)$, for each $j \in [1, m]$.

\end{itemize}
Thus,
\begin{align*}
\trans = & \big\{(\state_\init, r_j), (r_j, v^j_1), (u_n^j, 
u_n^j) \mid j \in [0, m] \big\} \; \cup\\
& \big\{(r_0, s_1), (r_0, \bar{s}_1) \big\} \; \cup\\
& \big\{(v^j_i, u^j_i) \mid i \in [1, n] \, \wedge \, j \in [1, m] \big\} \; 
\cup\\
& \big\{(u^j_i, v^j_{i+1}) \mid  i \in [1, n) \, \wedge \, j \in [1, m] \big\} 
\; \cup\\
& \big\{(s_i, \hat{s}_i), (\bar{s}_i, \hat{s}_i) \mid i \in [1, n] \big\} \; 
\cup\\
& \big\{(\hat{s}_i, s_{i+1}), (\hat{s}_i, \bar{s}_{i+1}) \mid i \in [1, n) 
\big\}.
\end{align*}

\end{itemize}


\begin{figure*}[t]
\centering
\includegraphics[scale=.95]{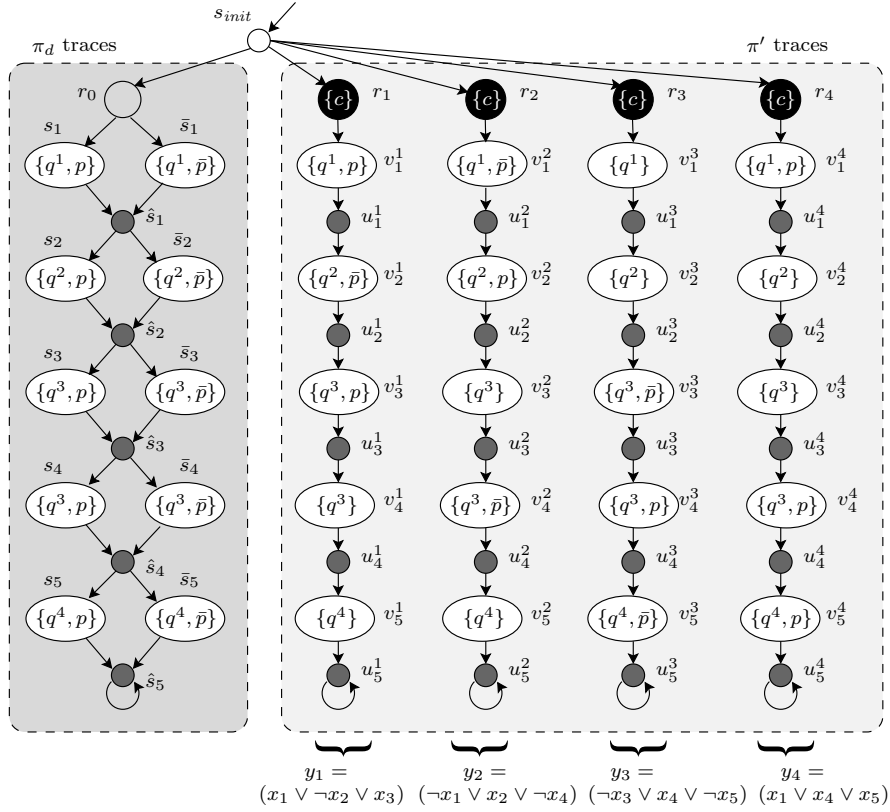}
\caption{Mapping quantified Boolean formula $y = \exists x_1.\forall 
x_2.\exists x_3.\exists x_4.\forall x_5.(x_1 \vee \neg x_2 \vee x_3) \wedge 
(\neg x_1 \vee x_2 \vee \neg x_4) \wedge (\neg x_3 \vee x_4 \vee \neg x_5) 
\wedge (x_1 \vee x_4 \vee x_5)$ to an instance of 
\MC{\mbox{(EA)}$k$\mbox{-HyperLTL}}{\mbox{acyclic}}.}
\label{fig:mapping-sys}
\end{figure*}

\newcommand{\map}{\mathit{map}}

\noindent \textbf{HyperLTL formula: } The HyperLTL formula in our mapping is 
the following:
\[ \begin{array}{l}
\label{eq:hltlsys}
\varphi_{\map} = \exists \pi_{k+1}.\forall \pi_{k} \cdots \exists \pi_2. 
\forall \pi_1. \forall \pi'. \\
\qquad  \Bigg( \bigwedge_{d \in \{1, 3, \dots, k\}} \X 
\neg c_{\pi_d} \, \wedge \, \X c_{\pi'}\Bigg) \; \Rightarrow \\
\nonumber \qquad \Bigg(\bigwedge_{d \in \{2, 4, \dots, k+1\}} \X \neg c_{\pi_d} 
 \; \wedge \\
\qquad ~~~~\F\bigg[\bigvee_{d \in [1,k+1]} \Big(\big(q^{d}_{\pi_d} 
\Leftrightarrow q_{\pi'}^d\big) \, \wedge \\
\qquad~~~~~~~~~~~~~~~ \big((p_{\pi'} \wedge p_{\pi_d}) \; 
\vee \; (\bar{p}_{\pi'} \wedge \bar{p}_{\pi_d})\big)\Big) 
\bigg]\Bigg)
\end{array}\]

\begin{figure}
\centering
\includegraphics[scale=.9]{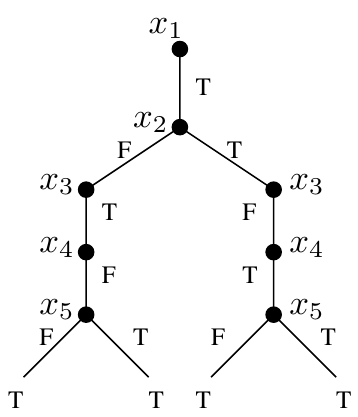}
\caption{Model for the QBF instance in Fig.~\ref{fig:mapping-sys}.}
\label{fig:qbfmodel-sys}
\end{figure}

Note that the formula has $k$ alternations. Intuitively, this formula 
expresses the following: for all the (clause) traces that are universally 
quantified (i.e., the left side of the implication), there exist (clause) 
traces, where either $p$ or $\bar{p}$ eventually matches its counterpart 
position in any trace $\pi'$. The matching positions identify the assignments 
of Boolean variables in the corresponding clauses that make the QBF instance 
true. 


We now show that the given quantified Boolean formula is $\mathit{true}$ if and 
only if the Kripke structure obtained by our mapping satisfies the HyperLTL 
formula $\varphi_\map$. 

\begin{description}

\item[($\Rightarrow$)] Suppose that $y$ is true. Then, there is an instantiation 
of existentially quantified variables for each value of universally quantified 
variables, such that each clause $y_j$, where $j \in [1, m]$ becomes true (see 
Fig.~\ref{fig:qbfmodel-sys} for an example). We now use these instantiations to 
instantiate each $\exists \pi_{x_d}$ in HyperLTL formula $\varphi_\map$, where 
$d \in \{2, 4, \dots, k+1\}$ as follows. For each existentially quantified 
variable $x_i$, where $i \in [1, n]$, in depth $d \in [1,k+1]$, if $x_i = 
\tru$, we instantiate $\pi_d$ with a trace that includes state $s_i$. 
Otherwise, the trace will include state $\bar{s}_i$. We now show that this 
trace instantiation evaluates formula $\varphi_\map$ to true. Observe that 
the left side of the implication in the formula is basically distinguishing 
traces (i.e., clause traces $\pi'$, where $\Next c$ holds and traces 
corresponding to universal variables, where $\neg \Next c$ is true). Since each 
$y_j$ is true, for any instantiation of universal quantifiers, there is at 
least one literal in $y_j$ that is true. If this literal is of the form $x_i$, 
then we have $x_i = \tru$ and trace $\pi_d$ will include $s_i$, which is 
labeled by $p$ and $q^d$. Hence, the values of $p$ (respectively, $q^d$), in 
both $\pi_d$ and $\pi'$ instantiated by trace
$$\state_{\init}r_jv^j_1\cdots u^j_n$$
are eventually equal. If the literal in $y_j$ is of the form $\neg x_i$, 
then $x_i = \fals$ and, hence, some trace $\pi_d$ will include $\bar{s}_i$. 
Again, the values of $\bar{p}$ (respectively, $q^d$), in both $\pi_d$ and 
$\pi'$ are eventually equal. Finally, since all clauses are true, all traces 
$\pi'$ reach a state where the right side of the implication becomes true.

\item[($\Leftarrow$)] Suppose our mapped Kripke structure satisfies the
HyperLTL formula $\varphi_\map$. This means that for each trace $\pi'$ of 
the form $$\state_{\init}r_jv^j_1\cdots u^j_n,$$
there exists a state $u_i^j$, where the values of $q^d$ and either $p$ or 
$\bar{p}$ are eventually equal to their counterparts in some trace $\pi_d$. If 
this trace is existentially quantified and includes $s_i$, then we assign $x_i 
= \tru$ for the preceding quantifications. If the trace includes $\bar{s}_i$, 
then $x_i = \fals$. Observe that since in no state $p$ and $\bar{p}$ are 
simultaneously true and no trace includes both $s_i$ and $\bar{s}_i$, variable 
$x_i$ will have only one truth value. This way, a model similar to 
Fig.~\ref{fig:qbfmodel-sys} can be constructed. Similar to the forward 
direction, it is straightforward to see that this valuation makes every clause 
$y_j$ of the QBF instance true.

To establish the hardness for HyperLTL formulas where the first quantifier
is universal, we analogously map an instance of QBF with $k$ alternations and 
where $\quant_1 = \forall$ to the model checking problem of an acyclic Kripke 
structure and a HyperLTL formula that also begins with a universal quantifier.
This time, the HyperLTL formula has $k+1$ quantifier alternations, because
the inner-most quantifier is universal. We have thus reduced a
\comp{$\mathsf{\Pi^p_{k+1}}$-hard} problem to the model checking problem for HyperLTL 
formulas with $k+1$ quantifier alternations where the first quantifier
is universal. Hence, the model checking problem for formulas with $k$ quantifier 
alternations where the first quantifier is universal is 
\comp{$\mathsf{\Pi^p_{k}}$-hard}.
\qed
\end{description}
\end{proof}

An important case of Theorem~\ref{thrm:eak-flat} are formulas with a single 
quantifier alternation, i.e., $k = 1$. This class of formulas contains, for 
example, generalized noninference~\cite{McLean:1994:GeneralTheory}, which can 
be expressed as a $\forall \exists$ and generalized 
non-interference~\cite{m88}, which can be expressed as a $\forall\forall 
\exists$ HyperLTL formula~\cite{cfkmrs14}. According to the polynomial 
hierarchy, the model checking problem for acyclic graphs is \comp{NP-complete} for formulas 
of the form $\exists^+\forall^+\psi$ and \comp{coNP-complete} for formulas of 
the form $\forall^+\exists^+\psi$. 

It is worth noting that the special case of a single quantifier alternation 
consisting of a single existential and a single universal quantifier 
is already \comp{NP/coNP-complete} for acyclic graphs, but still in \comp{L} for 
trees. The intuitive reason is the repeated-diamonds structure in 
Fig.~\ref{fig:mapping-sys}, which is possible in acyclic graphs, but not in 
trees. This structure allows us to select multiple Boolean values 
with a single trace quantifier. 

Finally, Theorem~\ref{thrm:eak-flat} implies that the model checking problem for acyclic 
Kripke structures and HyperLTL formulas with an arbitrary number of quantifiers 
is in \comp{PSPACE}. Moreover, its proof of lower bound shows that the problem 
is at least as hard as QBF, making it \comp{PSPACE-hard}.

\begin{corollary}
\label{cor:aestar-s}
\MC{HyperLTL}{\mbox{acyclic}} is \comp{PSPACE-complete} in the size of the 
Kripke structure.
\end{corollary}

\section{Combined Complexity}
\label{sec:combined}

We now analyze the complexity of the model checking problem in the size of the 
\emph{combined} input, consisting of both the Kripke structure {\em and} the 
HyperLTL formula. Again, we separately focus on trees and acyclic graphs.

\subsection{Trees}

For tree-shaped Kripke structures, we first show that model checking is \emph{efficiently 
parallelizable} for two fragments: (1) the alternation-free fragment 
(Theorem~\ref{thrm:aftree-c}), and (2) formulas with one alternation consisting 
of a single universal and a single existential quantifier 
(Theorem~\ref{thrm:aetree-c}). We denote (2) as {\sf (AE/EA)-HyperLTL}. As we 
already noted, this model checking problem is particularly interesting because its 
complexity is significantly different for trees and acyclic graphs. This is 
again true for the combined complexity, which is in \comp{NC} for trees, but 
\comp{$\mathsf{\Sigma^p_{2}}$-complete} or 
\comp{$\mathsf{\Pi^p_{2}}$-complete}, depending on whether the leading 
quantifier is existential or universal, for acyclic graphs.



\begin{theorem}
\label{thrm:aftree-c}
\MC{\mbox{(A/E)}$k$\mbox{-HyperLTL}}{\mbox{tree}} is in \comp{NC}.
\end{theorem}

\begin{proof}
A decision problem is in \comp{NC}, if there exists a parallel algorithm that 
runs in time $O(\log^c n)$ with $O(n^{c'})$ processors for some constants $c$ 
and $c'$.

To verify an alternation-free formula with $k$ quantifiers, we consider all 
combinations of $k$ traces in the Kripke structure. Since $k$ is a constant and 
the number of traces is bounded by the number of states of the tree, there
is only a polynomial number of combinations. The evaluation of an individual 
combination corresponds to the evaluation of an LTL formula over a single trace,
which can be done in \comp{NC}~\cite{kf09}. We evaluate all combinations in 
parallel.  

For universal quantifiers, we then compute the conjunction over these results by 
evaluating a binary tree of conjunctions. The height of the tree is logarithmic 
in the number of combinations. Using a linear number of processors, the 
evaluation therefore is done in logarithmic time. Likewise, for existential 
quantifiers, we compute the disjunction over the results by evaluating a binary 
tree of disjunctions.
\qed
\end{proof}



\begin{theorem}
\label{thrm:aetree-c}
\MC{(AE/EA)-HyperLTL}{\mbox{tree}} is in \comp{NC}.
\end{theorem}

\begin{proof}
Analogously to the proof of Theorem~\ref{thrm:aftree-c}, we consider all pairs 
of traces in the Kripke structure. Since the number of traces is bounded by the 
number of states of the tree, the number of pairs is polynomial. The evaluation 
of an individual pair corresponds to the evaluation of an LTL formula over a 
single trace, which can be done in \comp{NC}~\cite{kf09}. We evaluate all pairs 
in parallel. If the formula is of the form $\forall \exists$, we then need to 
evaluate the conjunction over all first elements of the pair, and the 
disjunction over all second elements. This can be done by a binary tree,
where the upper part consists of conjunctions and the lower part consists of 
disjunctions. The height of the tree is logarithmic in the number of pairs. 
Using a linear number of processors, the evaluation is therefore done in 
logarithmic time. Likewise, if the formula is of the form $\exists \forall$, we 
compute the disjunction over the results by evaluating a binary tree, where the 
upper part consists of disjunctions and the lower part of conjunctions.
\qed  
\end{proof}

\begin{theorem}
\label{thrm:eak-tree}
\MC{(EA/AE)$k$-HyperLTL}{\mbox{tree}} is 
\comp{$\mathsf{\Sigma^p_{k+1}}$-complete} in the combined size of the Kripke 
structure and the formula, if the leading quantifier is 
existential and is \comp{$\mathsf{\Pi^p_{k+1}}$-complete} if the leading 
quantifier is universal.
\end{theorem}

\begin{proof}
Matching upper bounds are provided in the proof of Theorem~\ref{thrm:aektree-c} 
in the next subsection for the more general case of acyclic graphs.
We now show that the model checking problem is \comp{$\mathsf{\Sigma^p_{k}}$-hard} 
(respectively, \comp{$\mathsf{\Pi^p_{k}}$-hard}) via a reduction from QBF 
satisfiability, where the leading quantifier is existential (respectively, 
universal). In contrast to the proof of Theorem~\ref{thrm:eak-flat}, we do not 
assume a specific form of the Boolean formula.

Let the quantified Boolean formula consist of Boolean variables $\{x_1, x_2, 
\dots, x_n\}$, and a formula with $k$ alternations
$$y=\quant_1 x_1.\quant_1 x_2\dots\quant_{n-1} x_{n-1}.\quant_n x_n.\, \varphi$$
where each $\quant_i \in \{\forall, \exists\}$ ($i \in [1, n]$) and $\varphi$ is
an arbitrary Boolean formula over variables $\{x_1, \ldots, x_n\}$.
Satisfiability for QBF formulas of this type is complete for 
\comp{$\mathsf{\Sigma^p_{k+1}}$} if $\quant_1 = 
\exists$, and for $\mathsf{\Pi^p_{k+1}}$ if $\quant_1 = \forall$.

We reduce the satisfiability problem for a quantified Boolean formula to the model checking problem
for a HyperLTL formula with the same quantifier structure.

\begin{itemize}
\item \textbf{Kripke structure $\krip = \ktuple$. }  We use the
  simple Kripke structure shown in Fig.~\ref{fig:npc-c}, which
  contains two traces $\{\}\{x\}^\omega$ and $\{\}\{\}^\omega$.

\item \textbf{HyperLTL formula. } The HyperLTL formula in our mapping is 
the following:
\begin{equation}
\label{eq:hltl-com}
 \quant_1 \pi_{1}.\quant_1 \pi_{2}\dots\quant_{n-1} \pi_{{n-1}}.\quant_n 
\pi_n.\, \varphi'
\end{equation}
where $\varphi'$ is constructed from $\varphi$ by replacing every occurrence
of a variable $x_i$ in the Boolean formula with $\X x_{\pi_i}$ in the HyperLTL
formula.
\end{itemize}

\begin{figure}[t]
 \centering
 \includegraphics[scale=.95]{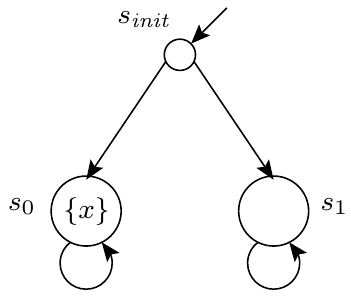}
 \caption{Kripke structure in the proof of Theorem~\ref{thrm:eak-tree}.}
 \label{fig:npc-c}
\end{figure}

The given formula is $\mathit{true}$ if and only if the Kripke structure 
obtained by our mapping satisfies HyperLTL formula~(\ref{eq:hltl-com}).
We translate every assignment to the trace quantifiers to a corresponding 
assignment of the Boolean variables, and vice versa, as follows: Assigning the 
trace $\{\}\{x\}^\omega$ to $\pi_i$ means that $x_i$ is set to $\tru$, 
and assigning the trace $\{\}\{\}^\omega$ to $\pi_i$ means that $x_i$ is set to 
$\fals$.
\qed
\end{proof}




\begin{corollary}
\label{cor:hltl-tree}
\MC{HyperLTL}{\mbox{tree}} is \comp{PSPACE-complete}.
\end{corollary}

\subsection{Acyclic Kripke Structures}



For HyperLTL formulas with bounded quantifier alternation, trees and acyclic 
graphs have the same model checking complexity (except for the special case of exactly 
one universal and one existential quantifier). We match the lower bounds for 
trees from Theorem~\ref{thrm:eak-tree} with upper bounds for acyclic graphs.

\begin{theorem}
\label{thrm:aektree-c}
\MC{(EA/AE)k-HyperLTL}{\mbox{acyclic}} is 
\comp{$\mathsf{\Sigma^p_{k+1}}$-complete} in the combined size of the Kripke 
structure and the formula, if the leading quantifier is 
existential and \comp{$\mathsf{\Pi^p_{k+1}}$-complete} if the leading 
quantifier is universal.
\end{theorem}

\begin{proof}
We show membership in \comp{$\mathsf{\Sigma^p_{k+1}}$} and 
\comp{$\mathsf{\Pi^p_{k+1}}$}, respectively, by induction over $k$. For the 
base case, $k = 0$, where  the formula is alternation-free, the model checking problem can 
be solved in \comp{NP} and \comp{co-NP}, respectively, as follows. If the 
quantifiers are existential, we can nondeterministically guess a combination of 
the traces and verify the correctness of the guess in polynomial time, as the 
length of each trace is bounded by the number of states. Likewise, if the 
quantifiers are universal, we can universally guess a combination of the traces 
and verify the correctness of the guess in polynomial time.

For $k+1$ quantifier alternations, suppose that the first quantifier is
existential. Since the Kripke structure is acyclic, the length of the traces is 
bounded by the number of states. We can thus nondeterministically guess the 
existentially quantified traces in polynomial time and verify the correctness 
of the guess, by the induction hypothesis, in \comp{$\mathsf{\Pi^p_{k+1}}$}. 
Hence, the model checking problem for $k+1$ is in \comp{$\mathsf{\Sigma^p_{k+2}}$}. 
Likewise, if the first quantifier is universal, we universally guess the 
universally quantified traces in polynomial time and verify the correctness of 
the guess, by the induction hypothesis, in \comp{$\mathsf{\Sigma^p_{k+1}}$}. 
Hence, the model checking problem for $k+1$ is in \comp{$\mathsf{\Pi^p_{k+2}}$}. 

Together with the lower bounds for trees in Theorem~\ref{thrm:eak-tree}, we 
obtain \comp{$\mathsf{\Sigma^p_{k+1}}$/$\mathsf{\Pi^p_{k+1}}$-completeness} for 
$k$ quantifier alternations.

\qed
\end{proof}


\begin{corollary}
\label{cor:hltl-acyc}
\MC{HyperLTL}{\mbox{acyclic}} is \comp{PSPACE-complete}.
\end{corollary}

\section{Related Work}
\label{sec:related}

Model checking algorithms for HyperLTL were introduced 
in~\cite{frs15}. The satisfiability problem for HyperLTL was shown to be 
decidable for the $\exists^*\forall^*$ 
fragment~\cite{fh16}. Runtime verification algorithms 
for HyperLTL include both automata-based algorithms~\cite{ab16,fhst17} and 
rewriting-based algorithms~\cite{bsb17}. HyperLTL is also supported by a growing 
set of tools, including the model checker MCHyper~\cite{frs15}, and the decision 
procedure EAHyper~\cite{fhs17}, and the runtime monitoring tool RVHyper~\cite{10.1007/978-3-319-89963-3_11}.

A study of the impact of structural restrictions on the complexity of
the model checking problem, similar to this paper, has been carried out
for LTL~\cite{kb11}.  The LTL model checking problem is \comp{PSPACE-hard} 
if there exists a strongly connected component with two distinct cycles in the 
Kripke structure. If no such component exists, then the 
model checking problem is in \comp{coNP}. For the special case
of finite paths and trees, the LTL model checking problem is
in \comp{NC}, or, more precisely, in \comp{AC$^1$(logDCFL)}~\cite{kf09,lars}.


\section{Conclusion}
\label{sec:conclusion}

We have developed a detailed and fundamental classification of the complexity 
of the model checking  problem for hyperproperties expressed in 
HyperLTL over trace logs that are stored as tree-shaped or acyclic Kripke 
structures. The complexity analysis is a crucial 
step for the development of runtime monitors, because in runtime verification methods for hyperproperties, the traces generated over time by the running system have to be 
stored into a growing data structure. This is a fundamental difference to monitoring techniques for standard trace properties, where the traces are evaluated individually and the monitors 
are usually memoryless.

We showed that for trees, the model checking complexity in the size of the Kripke structure 
is \comp{L-complete} independently of the number of quantifier alternations. For 
acyclic Kripke structures, the complexity is in \comp{PSPACE} (in the level of 
the polynomial hierarchy that corresponds to the number of quantifier 
alternations). The combined complexity in the size of the Kripke structure and 
the length of the HyperLTL formula is in \comp{PSPACE} for both trees and 
acyclic Kripke structures, and is as low as \comp{NC} for the relevant case of 
trees and alternation-free HyperLTL formulas.

These results highlight two crucial design choices for monitoring algorithms:

\begin{itemize}

\item The substantial differences between the complexities reported in 
Tables~\ref{tab:system} and~\ref{tab:combined}, in particular the contrast to 
the non-elementary complexity of the model checking problem for general graphs, 
are intriguing. These results suggest that non-exhaustive techniques such as 
runtime verification, that work on restricted structures, may have a significant complexity 
advantage over static verification. 

\item In the context of runtime verification, our results in Tables~\ref{tab:system} 
and~\ref{tab:combined} clearly show the tradeoffs in deploying runtime verification technology in 
practice. First, note that for runtime verification, the size of the formula is expected to 
remain constant and, hence, what matters is the size of the Kripke 
structure. Tables~\ref{tab:system} shows that the model checking complexity remains the 
same for trees, while it grows significantly for acyclic structures. This 
justifies careful space vs. time considerations in practical settings.

\end{itemize}

Our study raises many open questions for future work. An immediate question 
left unanswered in this paper is the precise complexity for trees and 
alternation-free HyperLTL formulas within \comp{NC}. Next, it would be 
interesting to determine the complexity of the verification problem for further 
restricted structures such as flat graphs, i.e., graphs that have no nested 
cycles. Also, there are many extensions of HyperLTL, such as the 
branching-time logic HyperCTL$^*$~\cite{cfkmrs14} and the first-order extension
FOHLTL~\cite{FinkbeinerMSZ-CCS17}.  It would be very interesting to see if the 
differences we observed for HyperLTL carry over to these much more expressive 
logics. And, finally, we are currently working on designing runtime verification techniques that 
can reuse the result of past verification steps as the size of the Kripke structure grows.

\paragraph{Acknowledgements}

This work was partially supported by Canada  NSERC
Discovery  Grant  418396-2012, by NSERC  Strategic  Grants
430575-2012 and 463324-2014, by the German Research Foundation (DFG) as part of the Collaborative Research Center ``Methods and Tools for Understanding and Controlling Privacy'' (SFB 1223), and by the European Research Council (ERC) Grant OSARES (No.\ 683300).

\bibliographystyle{IEEEtran}
\bibliography{bibliography}

\end{document}